\documentclass[11pt]{article}
\usepackage{wide, amsmath, amssymb, sariel, graphicx, euscript, ifthen}
\usepackage{times}

\def\tt{{\text{\boldmath $t$}}}

\def\bb{b}
\def\tvec{c}
\def\amat{M}

\def\qq{q}

\def\uu{u}

\def\zz{{\text{\boldmath $z$}}}

\def\llambda{{\text{\boldmath $\lambda$}}}

\def\pred{\mathop{\text{\sc Predecessor}}}
\def\succ{\mathop{\text{\sc Sucessor}}}
\def\TT{\EuScript{T}}
\def\MM{\mathbb{M}}

\def\reals{\mathbb{R}}

\def\zz{z}

\def\eval{\mathop{\mbox{\textsc{Evaluate}}}}
\def\evali{\mathop{\mbox{\textsc{Evaluate}}}^{-1}}
\def\ins{\mathop{\mbox{\textsc{Insert}}}}
\def\delete{\mathop{\mbox{\textsc{Delete}}}}
\def\atrans{\mathop{\textsc{Affine}}}
\def\itrans{\mathop{\textsc{Interval}}}

\def\Update{\mathop{\textsc{Update}}}
\def\UnUpdate{\mathop{\textsc{unUpdate}}}
\def\integrate{\mathop{\textsc{Integrate}}}

\def\merge{\mathop{\textsc{Merge}}}
\def\include{\mathop{\textsc{Include}}}
\def\unmerge{\mathop{\textsc{unMerge}}}
\def\uninclude{\mathop{\textsc{unInclude}}}
\def\pred{\mathop{\textsc{Pred}}}
\def\succ{\mathop{\textsc{Succ}}}

\newcommand\myCaption[1]{\small\refstepcounter{figure}%
   \figurename\ \thefigure :\ #1}

\newcommand{\lir}{LIR }
\newcommand{\ir}{IR }
\newcommand{\ur}{UR }
\newcommand{\hF}{\hat{F}}
\renewcommand{\c}[1]{\ensuremath{\EuScript{#1}}}
\renewcommand{\b}[1]{\ensuremath{\mathbb{#1}}}

\def\argmin{\mathop{\text{argmin}}}

\title{Lipschitz Unimodal and Isotonic Regression on Paths and Trees\footnote{The work was primarily done when the second and third authors were at Duke University. Research supported by NSF under grants
    CNS-05-40347, CFF-06-35000, and DEB-04-25465, by ARO grants
    W911NF-04-1-0278 and W911NF-07-1-0376, by an NIH grant
    1P50-GM-08183-01, by a DOE grant OEG-P200A070505, by a grant
    from the U.S.--Israel Binational Science Foundation, and a subaward to the University of Utah under NSF Award 0937060 to Computing Research Association.}}
\author{%
	\begin{tabular}{c}
	Pankaj K. Agarwal\\
	\small Duke University\\
	\small Durham, NC 27708\\
	\small pankaj@cs.duke.edu
	\end{tabular}
	\and 
	\begin{tabular}{c}
	Jeff M. Phillips\\
	\small University of Utah\\
	\small Salt Lake City, UT 84112\\
	\small jeffp@cs.utah.edu
	\end{tabular}
	\and 
	\begin{tabular}{c}
	Bardia Sadri\\
	\small University of Toronto\\
	\small Toronto, ON M5S 3G4\\
	\small sadri@cs.toronto.edu
	\end{tabular}
}%

\begin{document}
\begin{titlepage}
\maketitle
\thispagestyle{empty}
\begin{abstract}
We describe algorithms for finding the regression of $t$, a sequence of values, to the closest sequence $s$ by mean squared error, so that $s$ is always increasing (isotonicity) and so the values of two consecutive points do not increase by too much (Lipschitz).  The isotonicity constraint can be replaced with a unimodular constraint, where there is exactly one local maximum in $s$.  These algorithm are generalized from sequences of values to trees of values.  For each scenario we describe near-linear time algorithms.  
\end{abstract}
\end{titlepage}

\section{Introduction}\label{sec:intro}

Let $\MM$ be a triangulation of a polygonal region $P \subseteq \reals^2$ in which each vertex is 
associated with a real valued height (or elevation). Linear interpolation of vertex heights in the interior of 
each triangle of $\MM$ defines a piecewise-linear
function $t : P \rightarrow \reals$, called a \emph{height function}. 
A height function (or its graph) is widely used to model a 
two-dimensional surface in numerous applications (e.g. modeling the 
terrain of a geographical area).  With recent advances in sensing
and mapping technologies, such models are being generated
at an unprecedentedly large scale. These models are then used to 
analyze the surface and to compute various geometric and 
topological properties of the surface. For example, researchers in GIS
are interested in extracting river networks or computing visibility 
or shortest-path maps on terrains modeled as height functions. These structures depend heavily on 
the topology of the level sets of the surface and in particular on the topological relationship between its critical 
points (maxima, minima, and saddle points). 
Because of various factors such as measurement or sampling errors or the nature of the sampled surface, 
there is often plenty of noise in these surface models which introduces spurious 
critical points. This in turn leads to misleading or 
undesirable artifacts in the computed structures, e.g.,
artificial breaks in river networks. These difficulties have motivated
extensive work on topological simplification and noise removal through modification of the height function $t$ into another one 
$s: \MM \rightarrow \reals$ that has the desired set 
of critical points and that is as close to $t$ as 
possible~\cite{Bremer:03,Ni:04,Soille:04,Soille:04a,vKS09}. A popular
approach is to decompose the surface into pieces and modify each
piece so that it has a unique minimum or maximum~\cite{Soille:03a}.  In some applications, it is also desirable to impose the additional constraint that the function $s$ is Lipschitz; see below for further discussion. 

\paragraph{Problem statement.} 

Let $\MM = (V, A)$ be a planar graph with vertex set $V$ and arc (edge) set $A \subseteq V \times V$. We may treat $\MM$ as undirected in which case we take the pairs $(u, v)$ and $(v, u)$ as both representing the same undirected edge connecting $u$ and $v$.
Let $\gamma \geq 0$ be a real
parameter. A function $s: V \rightarrow \mathbb{R}$ is called
\begin{itemize} 
\item [(L)]{\em $\gamma$-Lipschitz} if  $(u, v) \in A$ implies $s(v) - s(u) \leq \gamma$.
\end{itemize}
Note that if $\MM$ is undirected, then Lipschitz constraint on an edge $(u,v) \in A$ implies $|s(u)-s(v)| \le \gamma$.  
For an undirected planar graph $\MM=(V,A)$, a function 
$s: V \rightarrow \reals$ is called
\begin{itemize}
\item [(U)]{\em unimodal} if $s$ has a unique {\em local maximum}, i.e. only one vertex $v \in V$ such that $s(v) > s(u)$ for all $(u, v) \in A$.
\end{itemize}
For a directed planar graph $\MM=(V,A)$, a function 
$s: V \rightarrow \reals$ is called
\begin{itemize}
\item [(I)]{\em isotonic} if $(u, v) \in A$ implies 
$s(u) \leq s(v)$.\footnote{%
A function $s$ satisfying the isotonicity constraint (I) must assign 
the same value to all the vertices of a directed cycle of $\MM$ 
(and indeed to all vertices in the same strongly connected component). 
Therefore, without loss of generality, we assume 
$\MM$ to be a directed {\em acyclic} graph (DAG).%
}
\end{itemize}
For an arbitrary function $t: V \rightarrow \mathbb{R}$ and a 
parameter $\gamma$, the {\em $\gamma$-Lipschitz unimodal regression 
($\gamma$-LUR\/)} of $t$ is a function $s: V \rightarrow \mathbb{R}$ 
that is $\gamma$-Lipschitz and unimodal on $\MM$ and 
minimizes $\|s - t\|^2  = \sum_{v \in V} (s(v) - t(v))^2$. 
Similarly, if $\MM$ is a directed planar graph, then 
$s$ is the {\em $\gamma$-Lipschitz isotonic regression 
($\gamma$-LIR\/)} of $t$ if $s$ satisfies (L) and (I) and 
minimizes $\|s - t\|^2$.
The more commonly studied {\em isotonic regression (IR)\/} and \emph{unimodal
regression}~\cite{AHKW06,ABERS55,BBBB72,Preparata:85} are the special cases of 
LIR and LUR, respectively, in which  $\gamma = \infty$, and therefore only the 
condition (I) or (U) is enforced.

Given a planar graph $\MM$, a parameter $\gamma$, and 
$t: V \rightarrow \mathbb{R}$, the LIR (resp.\ LUR) problem
is to compute the $\gamma$-LIR (resp.\ $\gamma$-LUR) of $t$. In this 
paper we propose near-linear-time algorithms for the LIR and LUR problems
for two special cases: when $\MM$ is a path or a tree. We study the special case where $\MM$ is a path prior to the more general case where it is tree because of the difference in running time and because doing so simplifies the exposition to the more general case. 

\paragraph{Related work.}
As mentioned above, there is extensive work on simplifying the topology
of a height function while preserving the geometry as much as possible.
Two widely used approaches in GIS are the so-called \emph{flooding} and
\emph{carving} techniques~\cite{Agarwal:06,Danner:07,Soille:04a,Soille:04}. 
The former technique raises the height of 
the vertices in ``shallow pits'' to simulate the effect of flooding,
while the latter lowers the value of the height function along a 
path connecting two pits so that the values along the path 
vary monotonically. As a result, one pit drains to the other and 
thus one of the minima ceases to exist. 
Various methods based on Laplacian smoothing have been proposed in the 
geometric modeling community to remove unnecessary critical 
points; see \cite{Bajaj:98,Bremer:03,Guskov:01,Ni:04} and references therein.
For example, Ni~\etal~\cite{Ni:04} proposed
the so-called \emph{Morse-fairing} technique, which solves a 
relaxed form of Laplace's equation, to remove undesired 
critical points.

A prominent line of research on topological simplification was 
initiated by Edelsbrunner~\etal~\cite{Edelsbrunner:00,Edelsbrunner:03} who introduced the 
notion of \emph{persistence}; see also~\cite{Edelsbrunner:07,Zomorodian:05,Zomorodian:09}. 
Roughly speaking,  each homology class of the 
contours in sublevel sets of a height function is characterized by 
two critical points at one of whom the class is born and at the other it is destroyed. The persistence of this class is then the height difference between these two critical points and can be thought of as the life span of 
that class. The persistence of a class effectively suggests the ``significance'' of its defining critical points. 
Efficient algorithms have been developed for computing the persistence
associated to the critical points of a height function and for simplifying topology based on 
persistence~\cite{Edelsbrunner:03,Bremer:03}.
Edelsbrunner~\etal~\cite{Edelsbrunner:06} and Attali~\etal~\cite{Attali:08} proposed algorithms 
for optimally eliminating all critical points of persistence below a 
threshold where the error is measured as 
$\|s-t\|_\infty = \max_{v \in V} |s(v)-t(v)|$. No efficient 
algorithm is known to minimize $\|s-t\|_2$.

The isotonic-regression (IR)  problem has been studied in 
statistics \cite{ABERS55,BBBB72,Preparata:85} since the 1950s.  It has many applications ranging from 
statistics \cite{Sto00} to bioinformatics \cite{BBBB72}, and from operations 
research \cite{MM85}  to differential optimization \cite{GW84}.  
Ayer \emph{et. al.} \cite{ABERS55} famously solves the \ir problem on 
paths in $O(n)$ time using the 
\emph{pool adjacent violator algorithm} (PAVA).  The algorithm 
works by initially treating each vertex as a level set and merging consecutive level sets that are out of order.  This algorithm is correct regardless of the order of the merges~\cite{RWD88}.  
Brunk \cite{Bru55} and Thompson \cite{Tho62} initiated the study of the \ir problem on general DAGs and trees, respectively.
Motivated by the problem of finding an optimal insurance rate structure over given risk classes for which a desired pattern of tariffs can be specified, Jewel \cite{Jewel:75} introduced the problem of Lipschitz isotonic regression on DAGs and showed connections between this problem and network flow with quadratic cost functions\footnote{It may at first seem that Jewel's formulation of the LIR $s$ of an input function $t$ on the vertices of a DAG is more general than ours in that he requires that for each $e = (u,v)\in A$, $s(v) \in [s(u) - \lambda_e, s(u)+\gamma_e]$ where $\lambda_e, \gamma_e \in \mathbb{R}^+$ are defined separately for each edge (in our formulation $\lambda_e = 0$ and $\gamma_e = \gamma$ for every edge $e$). Moreover, instead of minimizing the $L_2$ distance between $s$ and $t$ he requires $\sum_{v} w_v(s(v) - t(v))^2$ to be minimized where $w_v \in \mathbb{R}^+$ is a constant assigned to vertex $v$. However,  all our algorithms in this paper can be modified in a straight-forward manner to handle this formulation without any change in running-times.}.
Stout \cite{Sto08} solves the \ur problem on paths in $O(n)$ time.
Pardalos and Xue \cite{PX99} give an $O(n \log n)$ algorithm for the \ir problem on trees.  
For the special case when the tree is a star they give an $O(n)$ algorithm.  
Spouge \etal \cite{SWW03} give an $O(n^4)$ time algorithm for the \ir problem on DAGs.  
The problems can be solved under the $L_1$ and $L_{\infty}$ norms on 
paths \cite{Sto08} and DAGs \cite{AHKW06} as well, with an 
additional $\log n$ factor for $L_1$. Recently Stout~\cite{Stout:08} has
presented an efficient algorithm for istonotic regression in a planar
DAG under the $L_{\infty}$ norm. To our knowledge no polynomial-time algorithm is known for the 
UR problem on planar graphs, and there is no prior attempt on achieving efficient algorithms for  
Lipschitz isotonic/unimodal regressions in the literature. 

\paragraph{Our results.} 
Although the LUR problem for planar graphs remains elusive, we present
efficient exact algorithms for LIR and LUR problems on two special cases of 
planar graphs: paths and trees.  
In particular, we present an $O(n \log n)$ algorithm for computing the LIR on a path of length $n$ (Section~\ref{sec:pathlir}), and an $O(n \log n)$ algorithm on a tree with $n$ nodes (Section~\ref{sec:treelir}). 
We present an $O(n\log^2 n)$ algorithm for computing the
LUR problem on a path of length $n$ (Section~\ref{sec:pathlur}). 
Our algorithm can be extended to solve the LUR problem on an unrooted tree in $O(n \log^3 n)$ time (Section~\ref{sec:treelur}). 
The LUR algorithm for a tree is particularly interesting because of 
its application in the aforementioned carving 
technique~\cite{Danner:07,Soille:04,Soille:03a}. 
The carving technique modifies the height function along a number 
of trees embedded on the terrain where the heights of the vertices 
of each tree are to be changed to vary monotonically towards 
a chosen ``root'' for that tree. In other words, to perform the 
carving, we need to solve the IR problem on each tree. The downside of 
doing so is that the optimal IR solution happens to be a step function 
along each path toward the root of the tree with potentially large 
jumps. Enforcing the Lipschitz condition prevents sharp jumps 
in function value and thus provides a more natural solution 
to the problem.

Section~\ref{sec:datastructure} presents a data structure, called {\em affine composition tree} (ACT), for maintaining a $xy$-monotone polygonal chain, which can be regarded as the graph of a monotone piecewise-linear function $F: \reals \rightarrow \reals$.  
Besides being crucial for our algorithms, ACT is interesting in its own right.
A special kind of binary search tree, an ACT supports a number of operations to query and update the chain, each taking $O(\log n)$ time.
Besides the classical insertion, deletion, and query (computing $F(x)$ or $F^{-1}(x)$ for a given $x \in \reals$),  one can apply an $\itrans$ operation that modifies a contiguous subchain provided that the chain  remains $x$-monotone after the transformation, i.e., it remains the graph of a monotone function.  


\section{Energy Functions}
On a discrete set $U$, a real valued function $s: U \rightarrow \mathbb{R}$ can be viewed as a point in the $|U|$-dimensional Euclidean space in which coordinates are indexed by the elements of $U$ and the component $s_u$ of $s$ associated to an element $u \in U$ is $s(u)$. We use the notation $\mathbb{R}^U$ to represent the set of all real-valued functions defined on $U$. 

Let $\MM = (V, A)$ be a directed acyclic graph on which we wish to compute $\gamma$-Lipschitz isotonic regression of an {\em input function} $t \in \mathbb{R}^V$. For any set of vertices $U \subseteq V$, let $\MM[U]$ denote the subgraph of $\MM$ induced by $U$, i.e. the graph $(U, A[U])$, where $A[U] = \{(u, v) \in A: u, v \in U\}$.  The set of $\gamma$-Lipschitz isotonic functions on the subgraph $\MM[U]$ of $\MM$ constitutes
a convex subset of $\mathbb{R}^U$, denoted by $\Gamma(\MM,U)$.  It is the common intersection of all half-spaces determined by the isotonicity and Lipschitz constraints associated 
with the edges in $A[U]$, i.e., $s_u \le s_v$ and $s_v \le s_u+\gamma$ 
for all $(u,v) \in A[U]$.  

For $U \subseteq V$ we define $E_U : \mathbb{R}^V \to \mathbb{R}$ as 
\[
E_U(s) =  \sum_{v \in U} (s(v) - t(v))^2.
\]
Thus, the $\gamma$-Lipschitz isotonic regression of the input function $t$ is $\sigma = \arg\min_{s \in \Gamma(\MM,V)} E_V(s)$.  For a subset $U \subseteq V$ and $v \in U$ define the function $E_{\MM[U]}^v: \mathbb{R} \rightarrow \mathbb{R}$ as 
\[
E_{\MM[U]}^v(x) = \min_{s \in \Gamma(\MM,U); s(v) = x} E_U(s).
\]

\begin{lemma}\label{lem:convex-general}
For any $U \subseteq V$ and $v \in U$, the function $E_{\MM[U]}^v$ is continuous and strictly convex. 
\end{lemma}
%
\begin{proof}
Given $x, y \in \b{R}$ and $0 \leq \alpha \leq 1$, consider the functions 
\begin{eqnarray*}
z_x &=& \arg \min_{z \in \Gamma(\MM,U); z(v) = x} E_U(z)
\\
z_y &=& \arg \min_{z \in \Gamma(\MM,U); z(v) = y} E_U(z) 
\\
z_a &=& \alpha z_x + (1-\alpha) z_y.
\end{eqnarray*}
The strong convexity of $E_U$ implies that 
$$ 
\alpha E_U(z_x) + (1-\alpha) E_U(z_y) \geq E_U(z_\alpha).
$$
Furthermore, the convexity of $\Gamma(\MM,U)$ implies that $z_\alpha$ is also in $\Gamma(\MM,U)$.  Since by definition $z_\alpha(v) = \alpha x + (1-\alpha) y$, we deduce that 
$$
E_U(z_\alpha) \geq \min_{z \in \Gamma(\MM,u); z(v) = \alpha x + (1-\alpha)y} E_U(z).
$$
Thus, the minimum is precisely $E_U^v(\alpha x + (1-\alpha) y)$. 
This proves that $E_U^v$ is strictly convex on $\mathbb{R}$ and therefore continuous (and in fact differentiable except in countably many points).
%
%
\end{proof}

\section{Affine Composition Tree}\label{sec:datastructure}

In this section we introduce a data structure, called \emph{affine 
composition tree} (ACT), for representing an 
$xy$-monotone polygonal chain in $\reals^2$, which is being deformed 
dynamically. Such a chain can be regarded as the graph of a 
piecewise-linear monotone function $F:\reals\rightarrow \reals$, and thus is bijective. 
A {\em breakpoint} of $F$ is the $x$-coordinate of a {\em vertex} of 
the graph of $F$ (a vertex of $F$ for short), i.e., a 
$b \in \reals$ at which the left and right derivatives of $F$ disagree. The number of breakpoints of $F$ will be denoted by $|F|$. 
 A continuous piecewise-linear function $F$ with breakpoints 
$b_1 < \dots < b_n$ can be characterized by its vertices 
$\qq_i = (b_i, F(b_i)), i = 1, \dots, n$ together with the {\em slopes}
$\mu_-$ and $\mu_+$ of its {\em left} and {\em right unbounded pieces},
respectively, extending to $-\infty$ and $+\infty$. 
An {\em affine transformation} of $\reals^2$ is a map 
$\phi: \reals^2 \rightarrow \mathbb{R}^2,\, \qq \mapsto \amat 
\cdot \qq + \tvec$ where $\amat$ is a nonsingular $2\times 2$ matrix, 
(a {\em linear transformation\/}) and $\tvec \in \reals^2$ is a 
{\em translation vector} --- in our notation we treat 
$\qq \in \reals^2$ as a column vector.

\begin{figure}
\begin{center}
\includegraphics[scale=.6]{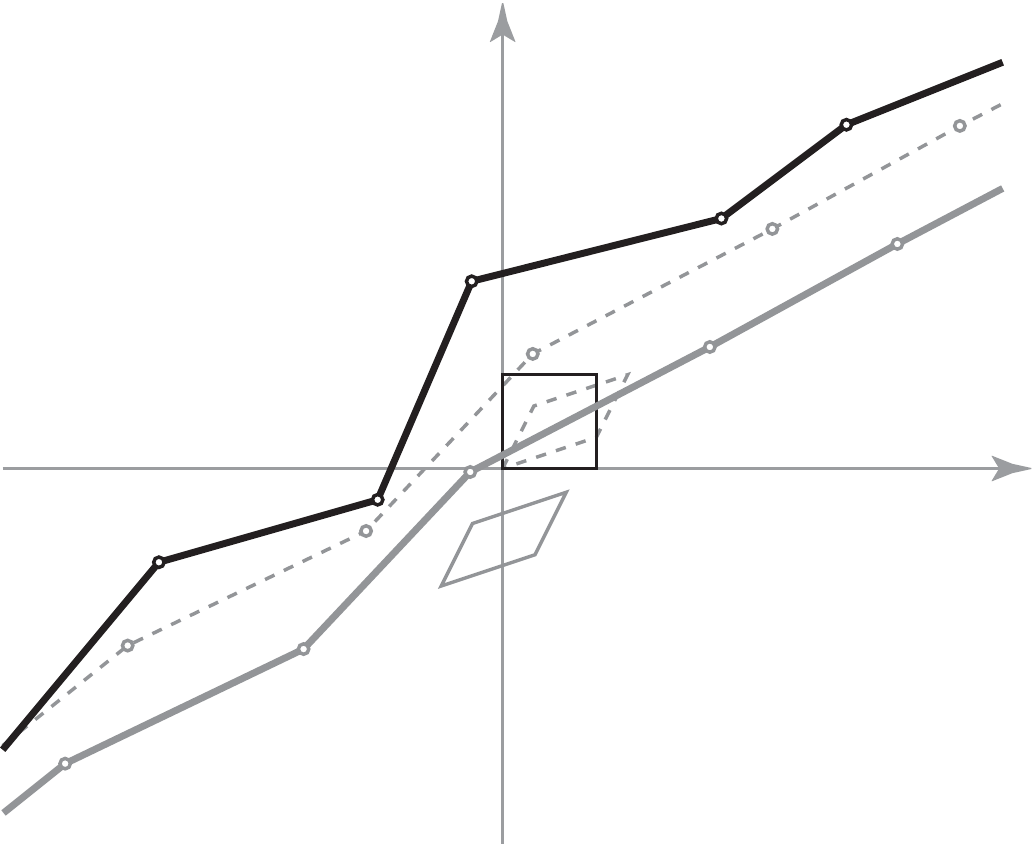}\qquad\qquad
\includegraphics[scale=.6]{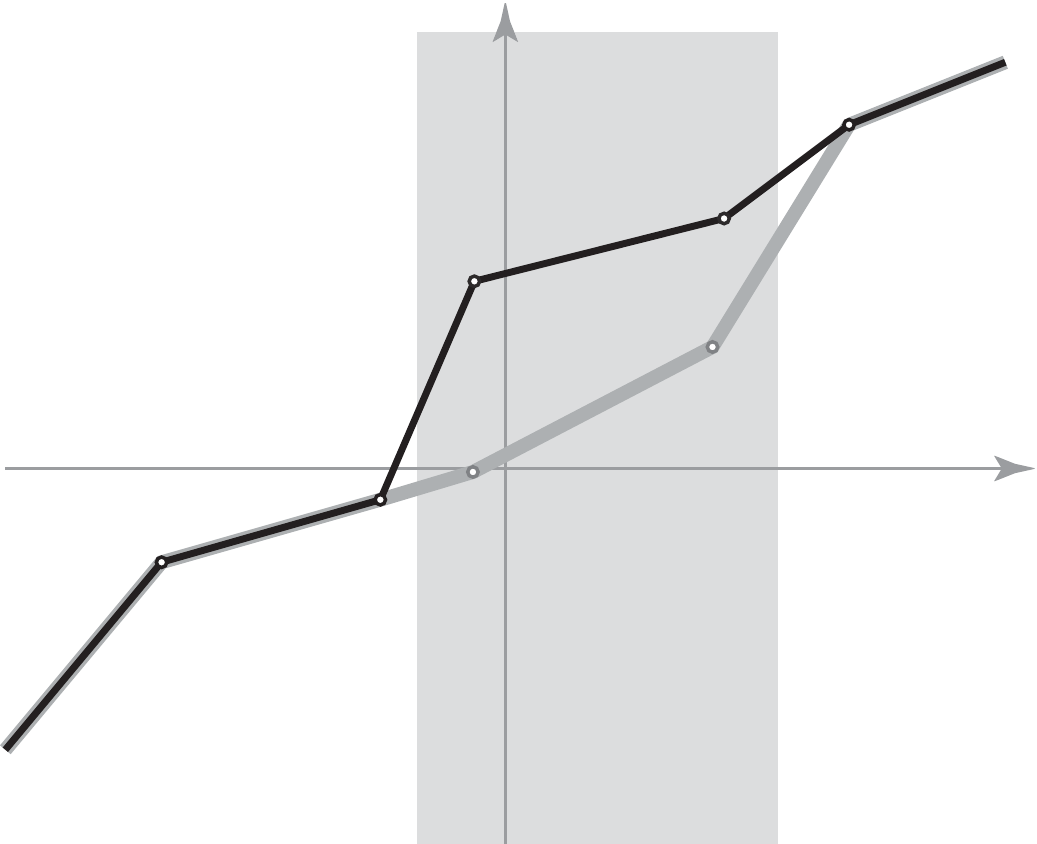}
\end{center}
\myCaption{\label{fig:trans} Left: the graph of a monotone piecewise 
linear function $F$ (in solid black). Vertices are marked by hollow 
circles. The dashed curve is the result of applying the linear 
transform $\psi_0(\qq) = \amat \qq$ where $\amat = 
\bigl(\protect\begin{smallmatrix}1 & 1/3 \\ 1/3 & 
	2/3\end{smallmatrix}\bigr)$. 
The gray curve, a translation of the dashed curve under vector 
$\tvec = \bigl(\begin{smallmatrix}-1 \\ -2\end{smallmatrix}\bigr)$, is the 
result of applying $\atrans(\psi)$ to $F$ where $\psi(\qq) = \amat \qq + \tvec$. 
Right: $\itrans(\psi, \tau^-, \tau^+)$, only the vertices of curve $F$ 
whose $x$-coordinates are in the marked interval [$\tau^-, \tau^+]$ 
are transformed. The resulting curve is the thick gray curve. }
\end{figure}

An ACT supports the following operations on a monotone piecewise-linear
function $F$ with vertices $\qq_i = (b_i, F(b_i)), i = 1, \dots, n$:
\begin{enumerate}
\item $\eval(a)$ and $\evali(a)$: Given any $a \in \mathbb{R}$, return 
$F(a)$ or $F^{-1}(a)$.

\item $\ins(\qq)$: Given a point $\qq = (x, y) \in \mathbb{R}^2$ insert 
$\qq$ as a new vertex of $F$. If $x \in (b_i, b_{i+1})$,
this operation removes the segment $\qq_i\qq_{i+1}$ from the graph of 
$F$ and replaces it with two segments $\qq_i\qq$ and $\qq\qq_{i+1}$, 
thus making $x$ a new breakpoint of $F$ with $F(x) = y$. 
If $x < b_1$ or $x > b_k$, then the affected unbounded piece of $F$ is 
replaced with one parallel to it but ending at $q$ and 
a segment connecting $q$ and the appropriate vertex of $F$ 
($\mu_+$ and $\mu_-$ remain intact).  We assume that 
$F(b_i) \le y \le F(b_{i+1})$.

\item $\delete(b)$: Given a breakpoint $b$ of $F$, removes the vertex 
$(b, F(b))$ from $F$; a delete operation modifies $F$ in a 
manner similar to insert.

\item $\atrans(\psi)$: Given an affine transformation 
$\psi:\mathbb{R}^2 \rightarrow \mathbb{R}^2$, modify the 
function $F$ to one whose graph is the result of application of 
$\psi$ to the graph of $F$. See Figure \ref{fig:trans}(Left).

\item $\itrans_p(\psi, \tau^-, \tau^+)$: Given an affine transformation
$\psi:\reals^2\rightarrow \reals^2$ and $\tau^-, \tau^+ \in \reals$ 
and $p \in \{x, y\}$, this operation applies $\psi$ to all vertices 
$v$ of $F$ whose $p$-coordinate is in the range $[\tau^-, \tau^+]$. Note that $\atrans(\psi)$ is equivalent to $\itrans_p(\psi, -\infty, +\infty)$ for $p = \{x,y\}$. See Figure \ref{fig:trans}(Right).


\end{enumerate} 

It must be noted that $\atrans$ and $\itrans$ are applied 
with appropriate choice of transformation parameters so that the 
resulting chain remains $xy$-monotone.  

An ACT $\TT = \TT(F)$ is a red-black tree that stores the vertices
of $F$ in the sorted order, i.e., the $i$th node of $\TT$ is
associated with the $i$th vertex of $F$. However, instead of storing
the actual coordinates of a vertex, each node $\zz$ of $\TT$ stores
an affine transformation 
$\phi_z: \qq \mapsto \amat_\zz \cdot \qq + \tvec_\zz$. 
If $\zz_0, \dots, \zz_k = \zz$ is the path in $\TT$ from the root
$\zz_0$ to $\zz$, then let 
$
\Phi_\zz(\qq) = \phi_{\zz_0}(\phi_{\zz_1}(\dots \phi_{\zz_k}(\qq)\dots)).
$
Notice that $\Phi_\zz$ is also an affine transformation. 
The actual coordinates of the vertex associated with $\zz$ are $(q_x, q_y) = \Phi_\zz(\overline{0})$ where $\overline{0} = (0, 0)$.

Given a value $b \in \reals$ and $p \in \{x,y\}$, let $\pred_p(b)$ (resp.\ $\succ_p(b)$) denote the rightmost (resp.\ leftmost) vertex $q$ of $F$ such that the $p$-coordinate of  $q$ is at most (resp.\ least) $b$.
Using ACT $\TT$, $\pred_p(b)$ and $\succ_p(b)$ can be computed in $O(\log n)$ time by following a path in $\TT$, composing the affine transformations along the path, evaluating the result at $\overline{0}$, and comparing its $p$-coordinate with $b$.
We can answer $\eval(a)$ queries by  determining the vertices $\qq^- = \pred_x(a)$ and $\qq^+ = \succ_x(a)$ of $F$ immediately preceding and succeeding $a$ and interpolating $F$ linearly between $\qq^-$ and $\qq^+$; if $a < b_1$ (resp. $a > b_k$), then $F(a)$ is calculated using $F(b_1)$ and $\mu_-$ (resp. $F(\bb_k)$ and $\mu_+$). 
Since $\bb^-$ and $\bb^+$ can be computed in $O(\log n)$ time and the interpolation takes constant time, $\eval(a)$ is answered in time $O(\log n)$. Similarly, $\evali(a)$ is answered using $\pred_y(a)$ and $\succ_y(a)$.

A key observation of ACT is that a standard rotation on any edge of $\TT$ can be performed in $O(1)$ time by modifying the stored affine transformations in a constant number of nodes (see Figure \ref{fig:rotation}) based on the fact that an affine transformation $\phi: \qq \mapsto \amat \cdot \qq + \tvec$ has an inverse affine transformation $\phi^{-1}: \qq \mapsto \amat^{-1}\cdot (\qq - \tvec)$; provided that the matrix $\amat$ is invertible.
A point $\qq \in \mathbb{R}^2$ is inserted into $\TT$ by first computing the affine transformation $\Phi_\uu$ for the node $\uu$ that will be the parent of the leaf $\zz$ storing $\qq$.  To determine $\phi_\zz$ we solve, in constant time, the system of (two) linear equations $\Phi_\uu(\phi_\zz(\overline{0})) = \qq$.  The result is the translation vector $\tvec_\zz$.  The linear transformation $\amat_\zz$ can be chosen to be an arbitrary invertible linear transformation, but for simplicity, we set $\amat_\zz$ to the identity matrix. Deletion of a node is handled similarly. 

\begin{figure}[b]
\begin{center}
\includegraphics[scale=.4]{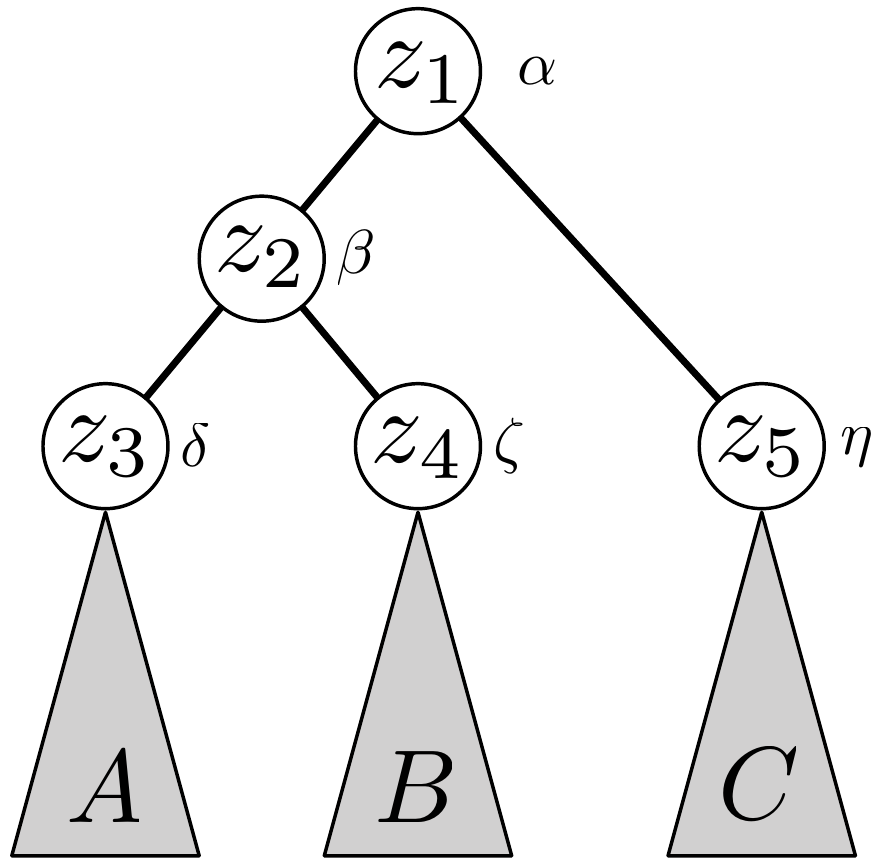}\qquad\qquad
\includegraphics[scale=.4]{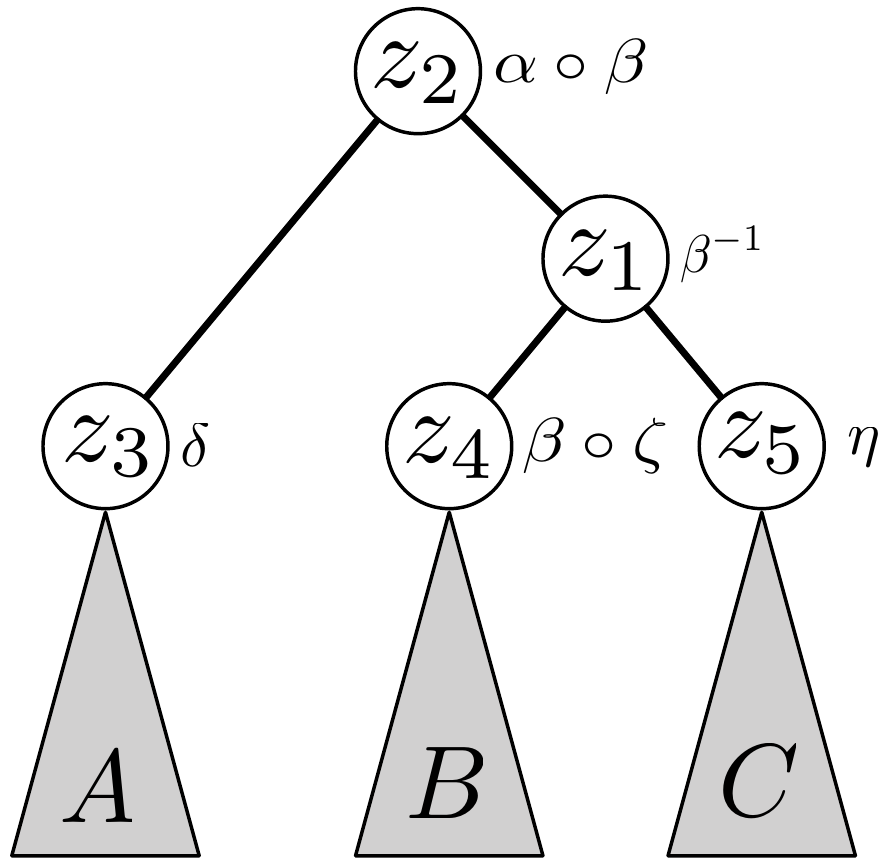}
\end{center}
\myCaption{\label{fig:rotation} Rotation in affine composition trees. The affine function stored at a node is shown as a greek letter to its right (left). When rotating the pair $(z_1, z_2)$, by changing these functions as shown (right) the set of values computed at the leaves remains unchanged. }
\end{figure}


To perform an $\itrans_p(\psi, \tau^-, \tau^+)$ query, we first find the nodes $\zz^-$ and $\zz^+$ storing the vertices $\pred_p(\tau^-)$ and $\succ_p(\tau^+)$, respectively. 
We then successively rotate $\zz^-$ with its parent until it becomes the root of the tree.  Next, we do the same with $\zz^+$ but stop when it becomes the right child of $\zz^-$.  At this stage, the subtree rooted at the left child of $\zz^+$ contains exactly all the vertices $\qq$ for which $q_p \in [\tau^-, \tau^+]$.  Thus we compose $\psi$ with the affine transformation at that node and issue the performed rotations in the reverse order to put the tree back in its original (balanced) position. Since $\zz^-$ and $\zz^+$ were both within $O(\log n)$ steps from the root of the tree, and since performing each rotation on the tree can only increase the depth of a node by one, $\zz^-$ is taken to the root in $O(\log n)$ steps and this increases the depth of $\zz^+$ by at most $O(\log n)$. Thus the whole operation takes $O(\log n)$ time.

We can augment $\TT(F)$ with additional information so that for any $a \in \b{R}$ the function $E(a) = E(b_1) + \int_{b_1}^a F(x) \, dx$, where $b_1$ is the leftmost breakpoint 
and $E(b_1)$ is value associated with $b_1$, 
can be computed in $O(\log n)$ time; we refer to this operation as $\integrate(a)$.  We provide the details in the Appendix.  
We summarize the above discussion:

\begin{theorem}\label{thm:ACT}
A continuous piecewise-linear monotonically increasing function 
$F$ with $n$ breakpoints can be maintained using a data structure 
$\TT(F)$ such that
\begin{enumerate}
\item $\eval$ and $\evali$ queries can be answered in $O(\log n)$ time,
\item an $\ins$ or a $\delete$ operation can be performed in $O(\log n)$ time, 
\item $\atrans$ and $\itrans$ operations can be performed in $O(1)$ and $O(\log n)$ time, respectively.
\item $\integrate$ operation can be performed in $O(\log n)$ time.  
\end{enumerate}
\end{theorem}

One can use the above operations to compute the sum of two increasing continuous piecewise-linear functions $F$ and $G$ as follows: we first compute $F(b_i)$ for every breakpoint $b_i$ of $G$ and insert the pair $(b_i, F(b_i))$ into $\TT(F)$.  At this point  the tree still represents $\TT(F)$ but includes all the breakpoints of $G$ as {\em degenerate} breakpoints (at which the left and right derivates of $F$ are the same). Finally, for every consecutive pair of breakpoints $b_i$ and $b_{i+1}$ of $G$ we  apply an $\itrans_x(\psi_i, b_i, b_{i+1})$ operation on $\TT(F)$ where $\psi_i$ is the affine transformation $\qq \mapsto \amat\qq + \tvec$ where
$
\amat = \bigl(\begin{smallmatrix}1 & 0 \\ \alpha & 1 \end{smallmatrix}\bigr)
$ 
and 
$
\tvec = \bigl(\begin{smallmatrix}0 \\ \beta \end{smallmatrix}\bigr),
$ 
in which $G_i(x) = \alpha x + \beta$ is the linear function that interpolates between $G(b_i)$ at $b_i$ and $G(b_{i+1})$ at $b_{i+1}$ (similar operation using $\mu_-$ and $\mu_+$ of $G$ for the unbounded pieces of $G$ must can be applied in constant time). It is easy to verify that after performing this series of $\itrans$'s, $\TT(F)$ turns into $\TT(F+G)$. The total running time of this operation is $O(|G| \log |F|)$. Note that this runtime can be reduced to $O(|G| (1+ \log|F|/\log |G|))$,  for $|G| < |F|$, by using an algorithm of Brown and Tarjan~\cite{BT79} to insert all breakpoints and then applying all $\itrans$ operations in a bottom up manner.  Furthermore, this process can be reversed (i.e. creating $\TT(F-G)$ without the breakpoints of $G$, given $\TT(F)$ and $\TT(G)$) in the same runtime.  We therefore have shown:

\begin{lemma}
Given $\TT(F)$ and $\TT(G)$ for a piecewise-linear isotonic functions $F$ and $G$ where $|G| < |F|$, $\TT(F+G)$ or $\TT(F-G)$ can computed in $O(|G|(1+\log |F| /\log |G|))$ time.
\label{lem:F+G}
\end{lemma}

\paragraph{Tree sets.}
In our application we will be repeatedly computing the sum of two functions $F$ and $G$.  It will be too expensive to compute $\TT(F + G)$ explicitly using Lemma \ref{lem:F+G}, therefore we represent it implicitly.  More precisely, we use a \emph{tree set} $\c{S}(F) = \{\TT(F_1), \ldots, \TT(F_k)\}$ consisting of affine composition trees of  monotone piecewise-linear functions $F_1, \ldots, F_k$ to represent the function $F = \sum_{j=1}^k F_j$. We perform several operations on $F$ or $\c{S}$ similar to those of a single affine composition tree.  

$\eval(x)$ on $F$ takes $O(k \log n)$ time, by evaluating $\sum_{j=1}^k F_j(x)$.  
$\evali(y)$ on $F$ takes $O(k \log^2 n)$ time using a technique of Frederickson and Johnson~\cite{Frederickson:82}.

Given the ACT $\TT(F_0)$, we can convert $\c{S}(F)$ to $\c{S}(F + F_0)$ in two ways:
an $\include(\c{S}, F_0)$ operation sets $\c{S} = \{\TT(F_1), \ldots, \TT(F_k), \TT(F_0)\}$ in $O(1)$ time.
A $\merge(\c{S}, F_0)$ operations sets $\c{S} = \{\TT(F_1 + F_0), \TT(F_2), \ldots, \TT(F_k)\}$ in $O(|F_0| \log |F_1|)$ time.
We can also perform the operation $\uninclude(\c{S}, F_0)$ and $\unmerge(\c{S}, F_0)$ operations that reverse the respective above operations in the same runtimes.  

We can perform an $\atrans(\c{S}, \psi)$ where $\psi$ describes a linear transform $\amat$ and a translation vector $\tvec$.  To update $F$ by $\psi$ we update $F_1$ by $\psi$ and for $j \in [2,k]$ update $F_j$ by just $\amat$.  This takes $O(k)$ time.  It follows that we can perform $\itrans(\c{S}, \psi, \tau^-, \tau^+)$ in $O(k \log n)$ time, where $n = |F_1| + \ldots + |F_k|$.  Here we assume that the transformation $\psi$ is such that each $F_i$ remains monotone after the transformation.

\section{LIR on Paths}\label{sec:pathlir}
In this section we describe an algorithm for solving the \lir problem on a path, represented as a directed graph $P = (V, A)$ where $V = \{v_1, \dots, v_n\}$ and $A = \{(v_i, v_{i+1}): 1 \leq i < n\}$.   A function $s: V \rightarrow \mathbb{R}$ is isotonic (on $P$) if $s(v_i) \leq s(v_{i+1})$, and $\gamma$-Lipschitz for some real constant $\gamma$ if $s(v_{i+1}) \leq s(v_i) + \gamma$ for each $i = 1, \dots, n - 1$. 
For the rest of this section let $t: V \rightarrow \mathbb{R}$ be an input function on $V$. 
For each $i = 1, \dots, n$, let $V_i = \{v_1, \dots, v_i\}$, let $P_i$ be the subpath $v_1, \ldots, v_i$, and let $E_i$ and $\tilde E_i$, respectively, be shorthands for $E_{V_i}$ and $E_{P_i}^{v_i}$.
By definition, if we let $\tilde E_0 = 0$, then for each $i \geq 1$: 
\begin{equation}\label{eq3}
\tilde E_{i}(x)  = (x - t(v_{i}))^2 + \min_{x - \gamma \leq y \leq x} \tilde E_{i-1}(y).
\end{equation} 
By Lemma \ref{lem:convex-general}, $\tilde E_i$ is convex and continuous and thus has a unique minimizer $s^*_i$.

\begin{lemma}\label{lem:convex}
For $i \geq 1$, the function $\tilde E_i$ is given by the recurrence relation:
\begin{equation}\label{eqei}
\tilde E_{i}(x) = (x - t(v_i))^2 + \left\{\begin{array}{lll} \tilde E_{i-1}(x - \gamma) & \qquad & x > s^*_{i-1} + \gamma \\
\tilde E_{i-1}(s^*_{i-1}) & & x \in [s^*_{i-1}, s^*_{i-1}+\gamma] \\
\tilde E_{i-1}(x)&& x < s^*_{i-1}. \end{array}\right.
\end{equation}
\end{lemma}
\begin{proof}
The proof is by induction on $i$. $\tilde E_1$ is clearly a single-piece quadratic function. For $i > 1$, since $\tilde E_{i-1}$ is strictly convex, it is strictly decreasing on $(-\infty, s^*_{i-1})$ and strictly increasing on $(s^*_{i-1}, +\infty)$. Thus depending on whether $s^*_{i-1} < x - \gamma$, $s^*_{i-1} \in [x - \gamma, x]$, or $s^*_{i-1} > x$, the value $y \in [x - \gamma, x]$ that minimizes $\tilde E_i(y)$ is $x-\gamma$, $s^*_i$, and $x$, respectively. 
\end{proof}

\begin{figure}
\begin{center}
\includegraphics[scale=.9]{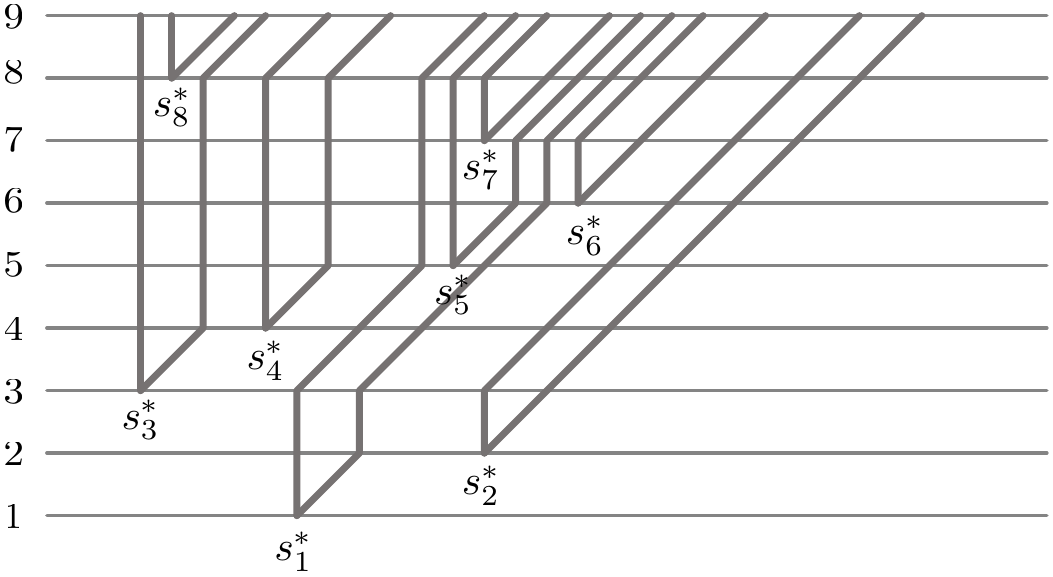}
\end{center}
\myCaption{\label{fig:LIR-breakpoints} The breakpoints of the function $F_i = d \tilde{E}_i / dx$. For each $i$, $s^*_{i-1}$ and $s^*_{i-1}+\gamma$ are the ``new'' breakpoints of $F_i$. All other breakpoints of $F_i$ come from $F_{i-1}$ where those smaller than $s^*_i$ remain unchanged and those larger are increased by $\gamma$.}
\end{figure}

Thus by Lemmas \ref{lem:convex-general} and \ref{lem:convex}, $\tilde E_i$ is strictly convex and piecewise quadratic. 
We call a value $x$ that determines the boundary of two neighboring quadratic pieces of $\tilde E_i$ a {\em breakpoint} of $\tilde E_i$. Since $\tilde E_1$ is a simple (one-piece) quadratic function, it has no breakpoints. For $i > 1$, the  breakpoints of the function $\tilde E_i$ consist of $s^*_{i-1}$ and $s^*_{i-1} + \gamma$, as determined by recurrence (\ref{eqei}), together with breakpoints that arise from recursive applications of $\tilde E_{i-1}$. 
Examining equation (\ref{eqei}) reveals that all breakpoints of $\tilde E_{i-1}$ that are smaller than $s^*_{i-1}$ remain breakpoints in $\tilde E_i$ and all those larger than $s^*_{i-1}$ are increased by $\gamma$ and these form all of the breakpoints of $\tilde E_i$ (see Figure \ref{fig:LIR-breakpoints}).  Thus $\tilde E_i$ has precisely $2i - 2$ breakpoints. 
To compute the point $s^*_i$ at which $\tilde E_i(x)$ is minimized, it is enough to scan over these $O(i)$ quadratic pieces and find the unique piece whose minimum lies between its two ending breakpoints. 
%

\begin{lemma}\label{lem:backsolve}
Given the sequence $s^*_1, \dots, s^*_n$, one can compute the $\gamma$-LIR $s$ of input function $t$ in $O(n)$ time. 
\end{lemma}
%
\begin{proof}
For each $i = 1, \dots, n$, let $\sigma_{i} = \argmin_{x \in \Gamma(P, {V_i})} E_i(x)$, i.e. $\sigma_i$ is the $\gamma$-LIR of $\tau_i \in \mathbb{R}^{V_i}$ where $\tau_i(v) = t(v)$ for all $v \in V_i$ and in particular $\sigma_n = s$ where $s$ is the $\gamma$-LIR of $t$. From the definitions we get \[E_i(\sigma_i) = \min_{x \in \mathbb{R}} \tilde E_i(x).\]
The uniqueness of the $\sigma_i$ and $s^*_i$ implies that $\sigma_i(v_i) = s^*_i$. In particular, $\sigma_n(v_n) = s(v_n) = s^*_n$. Suppose that the numbers $s(v_{i+1}), \dots, s(v_{n})$ are determined, for some $1 \leq i \leq n - 1$, and we wish to determine $s(v_i)$. Picking $s(v_i) = x$ entails that the energy of the solution will be at least 
\[
\tilde E_i(x) + \sum_{j = i+1}^n (s(v_j) - t(v_j))^2.
\]
Thus $s(v_i)$ has to be chosen to minimize $\tilde E_i$ on the interval $[s(v_{i+1}) - \gamma, s(v_{i+1})]$. If $s^*_i \in [s(v_{i+1}) - \gamma, s(v_{i+1})]$, then we simply have $s(v_i) = s^*_i$ which is the global minimum for $\tilde E_i$. Otherwise, one must pick $s(v_i) = s(v_{i+1}) - \gamma$ if $s^*_i < s(v_{i+1}) - \gamma$ and $s(v_i) = s(v_{i+1})$ if $s^*_i > s(v_{i+1})$. Thus each $s(v_i)$ can be found in $O(1)$ time for each $i$ and the entire process takes $O(n)$ time.
\end{proof}

One can compute the values of $s^*_1, \dots, s^*_n$ in $n$ iterations. The $i$th iteration computes the value $s^*_i$ at which $\tilde E_i$ is minimized and then uses it to compute the function $\tilde E_{i+1}$ as given by (\ref{eqei}) in $O(i)$ time.  After having computed $s_i^*$'s, the $\gamma$-LIR of $t \in \mathbb{R}^V$ can be computed in linear time.  
However, this gives an $O(n^2)$ algorithm for computing the $\gamma$-LIR of $t$. We now show how this running time can be reduced to $O(n \log n)$.


For the sake of simplicity, we assume for the rest of this paper that 
for each $2 \leq i \leq n$, $s^*_i$, i.e., the point minimizing 
$\tilde E_i$ is none of its breakpoints, although it is not hard to 
relax this assumption algorithmically. 
Under this assumption, $s^*_i$ belongs to the interior of some interval on which $\tilde E_i$ is quadratic. The derivative of this quadratic function is therefore zero at $s^*_i$. In other words, if we know to which quadratic piece of $\tilde E_i$ the point $s^*_i$ belongs, we can determine $s^*_i$ by setting the derivative of that piece to zero. 

\begin{lemma}\label{lem:monotone}
The derivative of $\tilde E_i$ is a continuous monotonically increasing piecewise-linear function. 
\end{lemma}
%
\begin{proof}
Since by Lemma \ref{lem:convex}, $\tilde E_i$ is strictly convex and piecewise quadratic, its derivative  is monotonically increasing. To prove the continuity of the derivative, it is enough to verify this at the ``new'' breakpoints of $\tilde E_i$, i.e. at $s^*_i$ and $s^*_i + \gamma$. The continuity of the derivative can be argued inductively by observing that both mappings of the breakpoints from $\tilde E_{i - 1}$ to $\tilde E_i$ (identity or shifting by $\gamma$) preserve the continuity of the function and its derivative. 

Notice that at $s^*_{i-1}$, both the left and right derivatives of a quadratic piece of $\tilde E_{i-1}$ becomes zero. From the definition of $\tilde E_i$ in (\ref{eqei}), the left derivative of the function $\tilde E_i - (x - t_{i+1})^2$ at $s^*_{i-1}$ agrees with the left derivative of $\tilde E_{i-1}$ at the same point and is therefore zero. On the other hand, the right derivative of the function $\tilde E_i - (x - t_{i+1})^2$ is zero at $s^*_{i-1} + \gamma$ and agrees with its  right derivative at the same point as the function is constant on the interval $[s^*_{i-1}, s^*_{i-1}+\gamma]$. This means that the derivative of $\tilde E_i - (x - t_{i+1})^2$ is continuous at $s^*_{i-1}$ and therefore the same holds for $\tilde E_i$. A similar argument establishes the continuity of $\tilde E_i$ at $s^*_{i-1} + \gamma$.
\end{proof} 

Let $F_i$ denote the derivative of $\tilde E_i$.
Using (\ref{eqei}), we can write the following recurrence for $F_i$:
\begin{equation}\label{eqfi}
F_{i+1} = 2(x-t(v_{i+1})) + \hF_i(x)
\end{equation}
where 
\begin{equation}\label{eq:hatF}
\hF_i(x) = 
\left\{\begin{array}{lll}F_i(x - \gamma) & \qquad & x > s^*_i + \gamma,\\
0 & & x \in [s^*_i, s^*_i + \gamma], \\
F_i(x) && x< s^*_i. \end{array} \right.
\end{equation}
if we set $F_0 = 0$. As mentioned above, $s^*_i$ is simply the solution of $F_i(x) = 0$, which, by Lemma \ref{lem:monotone} always exists and is unique. 
Intuitively, $\hF_i$ is obtained from $F_i$ by splitting it at $s_i^*$, shifting the right part by $\gamma$, and connecting the two pieces by a horizontal edge from $s^*_i$ to $s^*_i + \gamma$ (lying on the $x$-axis).  


In order to find $s^*_i$ efficiently, we use an ACT $\TT(F_i)$ to represent $F_i$.  It takes $O(\log |F_i|) = O(\log i)$ time to compute $s^*_i = \evali(0)$ on $\TT(F_i)$.  
Once $s^*_i$ is computed, we store it in a separate array for back-solving through Lemma~\ref{lem:backsolve}.  
We turn $\TT(F_i)$ into $\TT(\hF_{i})$ by performing a sequence of
$\ins((s_i^*,0))$, $\itrans_x(\psi, s_i^*, \infty)$, and $\ins((s_i^*,0))$ operations on $\TT(F_i)$ where 
$
\psi(q) = q + \bigl(\begin{smallmatrix} \gamma \\ 
		0\end{smallmatrix}\bigr);
$
the two insert operations add the breakpoints at $s_i^*$ and $s_i^*+\gamma$ and the interval operation shifts the portion of $F_i$ to the right of $s_i^*$ by $\gamma$.
We then turn $\TT(\hF_i)$ into $\TT(F_{i+1})$ by performing $\atrans(\phi_{i+1})$ operation on $\TT(\hF_i)$ where $\phi_{i+1}(q) = \amat q + \tvec$ where 
$\amat = \bigl(\begin{smallmatrix} 1 & 0 \\ 2 & 1\end{smallmatrix}\bigr)$ 
and 
$\tvec = \bigl(\begin{smallmatrix}0 \\ -2t(v_{i+1})\end{smallmatrix}\bigr)$.


Given ACT $\TT(F_i)$, $s_i^*$ and $\TT(F_{i+1})$ can be computed in $O(\log i)$ time.  Hence, we can compute $s_1^*, \ldots, s_n^*$ in $O(n \log n)$ time.  By Lemma \ref{lem:backsolve}, we can conclude the following.

\begin{theorem}
Given a path $P = (V, A)$, a function $t \in \mathbb{R}^V$, and a constant $\gamma$, the $\gamma$-Lipschitz isotonic regression of $t$ on $P$ can be found in $O(n \log n)$ time.
\end{theorem}

\paragraph{$\Update$ operation.}
We define a procedure $\Update(\TT(\hF_i), t(v_{i+1}), \gamma)$ that encapsulates the process of turning $\TT(\hF_i)$ into $\TT(\hF_{i+1})$ and returning $s_{i+1}^*$.  
Specifically, it performs $\atrans(\phi_{i+1})$ of $\TT(\hF_i)$ to produce $\TT(F_{i+1})$, then it outputs $s_{i+1}^* = \evali(0)$ on $\TT(F_{i+1})$, and finally a sequence of $\ins((s_{i+1}^*,0))$, $\itrans(\psi, s_{i+1}^*, \infty)$, 
and $\ins((s_{i+1}^*,0))$ operations on $\TT(F_{i+1})$ to 
get $\TT(\hF_{i+1})$.  
Performed on $\TT(F)$ where $F$ has $n$ breakpoints, an $\Update$ takes $O(\log n)$ time.  
An $\UnUpdate(\TT(\hF_{i+1}), t(v_{i+1}), \gamma)$ reverts the affects of an $\Update(\TT(\hF_i), t(v_{i+1}), \gamma)$.  This requires that $s_{i+1}^*$ is stored for the reverted version.  
Similarly, we can perform $\Update(\c{S}, t(v_i), \gamma)$ and $\UnUpdate(\c{S}, t(v_i), \gamma)$ on a tree set $\c{S}$, in $O(k \log^2 n)$ time, the bottleneck coming from $\evali(0)$.

\begin{lemma}\label{lem:tree-update}
Given $\TT(\hF_i)$, $\Update(\TT(\hF_i), t(v_i), \gamma)$ and $\UnUpdate(\TT(\hF_i), t(v_i), \gamma)$ take $O(\log n)$ time.
\end{lemma}
%
\begin{proof}
Since $F_{i+1}(x) = 2(x-t(v_{i+1})) + \hF_{i}(x)$, for a breakpoint $\hat{b}_{i}$ of $\hF_{i}$, we get:
\[
\left(\begin{array}{c}b_{i+1} \\ F_{i+1}(b_{i+1})\end{array}\right) = \left(\begin{array}{cc} 1 & 0 \\ 2 & 1 \end{array}\right)\left(\begin{array}{c} \hat{b}_{i} \\ \hF_{i}(\hat{b}_{i}) \end{array}\right) + 
\left(\begin{array}{c} 0 \\ - 2t(v_{i+1}) \end{array}\right), 
\]
where $\hat{b}_{i}$ becomes breakpoint $b_{i+1}$.

We can now compute $s^*_{i+1} = \evali(0)$ on $\TT(F_{i+1})$.

Rewriting (\ref{eq:hatF}) for a breakpoint $b_{i+1}$ of $F_{i+1}$ that is neither of $s^*_{i+1}$ or $s^*_{i+1} + \gamma$, using the fact that 
\begin{equation}\label{eqxi}
\hat{b}_{i+1} = \left\{\begin{array}{lll}b_{i+1} + \gamma & \qquad & b_{i+1} > s^*_{i+1}\\
b_{i+1} & &b_{i+1} < s^*_{i+1}
\end{array}\right.,
\end{equation}
where $b_{i+1}$ is the breakpoint of $F_{i+1}$ that has become $\hat{b}_{i+1}$ in $\hF_{i+1}$, we get:
\begin{equation}\label{eqfxi}
\hF_{i+1}(\hat{b}_{i+1}) =  \left\{\begin{array}{lll} F_{i+1}(b_{i+1}) & \qquad & b_{i+1} > s^*_{i+1} \\
F_{i+1}(b_{i+1}) && b_{i+1}< s^*_{i+1}. \end{array} \right.
\end{equation}
We next combine and rewrite equations (\ref{eqxi}) and (\ref{eqfxi}) as follows:
$$
\left(\begin{array}{c}\hat{b}_{i+1} \\ \hF_{i+1}(\hat{b}_{i+1})\end{array}\right) = \left(\begin{array}{cc} 1 & 0 \\ 0 & 1 \end{array}\right)\left(\begin{array}{c} b_{i+1} \\ F_{i+1}(b_{i+1}) \end{array}\right) + \left\{%
\begin{array}{lll}
\left(\begin{array}{c} \gamma \\ 0 \end{array}\right) & \qquad & b_{i+1} > s^*_{i+1} \\
\left(\begin{array}{c} 0 \\ 0 \end{array}\right) & \qquad & b_{i+1} < s^*_{i+1}
\end{array}\right.
$$

The new breakpoints, i.e. the pairs $(s^*_{i}, 0)$ and $(s^*_{i} + \gamma, 0)$, should be inserted into $\TT(F_{i})$. 

To perform $\UnUpdate(\TT(F_i),t(v_i), \gamma)$, we simply perform the inverse of all of these operations.  
\end{proof}

\section{LUR on Paths} \label{sec:pathlur}
Let $P = (V, A)$ be an {\em undirected} path where $V = \{v_1, \dots, v_n\}$ and $A = \{\{v_i,v_{i+1}\}, 1 \leq i < n\}$ and given $t \in \reals^V$. For $v_i \in V$ let $P_i=(V,A_i)$ be a directed graph in which all edges are directed 
towards $v_i$; that is, for $j<i$, $(v_j, v_{j+1}) \in A_i$ and for $j>i$ $(v_j, v_{j-1}) \in A_i$.  
For each $i = 1, \dots, n$, let $\Gamma_i = \Gamma(P_i, V) \subseteq \reals^V$ and let $\sigma_i = \arg\min_{s \in \Gamma_i} E_V(s)$.  
If $\kappa = \arg \min_{i} E_V(\sigma_i)$, then $\sigma_\kappa$ is the $\gamma$-LUR of $t$ on $P$.

We can find $\sigma_\kappa$ in $O(n \log^2 n)$ time by solving the LIR problem, then traversing the path while maintaining the optimal solution using $\Update$ and $\UnUpdate$.  Specifically,
for each $i=1,\dots, n$, let $V_i^- = \{v_1, \ldots, v_{i-1}\}$ and $V_i^+ = \{v_{i+1}, \ldots, v_n\}$. 
For $P_i$, let $\hF^-_{i-1}$, $\hF^+_{i+1}$ be the functions on directed paths $P[V_i^-]$ and $P[V_i^+]$, respectively, as defined in (\ref{eqfi}).  Set $\bar{F}_i(x) = 2(x-t(v_i))$.  Then the function $F_i(x) = d E_{v_i}(x) / dx$ can be written as $F_i(x) = \hF_{i-1}^-(x) + \hF_{i+1}^+(x) + \bar{F}_i(x)$.  We store $F_i$ as the tree set $S_i = \{\TT(\hF_{i-1}^-), \TT(\hF_{i+1}^+), \TT(\bar{F}_i)\}$.  
By performing $\eval^{-1}(0)$ we can compute $s_i^*$ in $O(\log^2 n)$ time (the rate limiting step), and then we can compute $E_{V}(s_i^*)$ in $O(\log n)$ time using $\integrate(s_i^*)$.  
Assuming we have $\TT(\hF_{i-1}^-)$ and $\TT(\hF_{i+1}^+)$, we can construct $\TT(\hF_i^-)$ and $\TT(\hF_{i+2}^+)$ in $O(\log n)$ time be performing 
$\Update(\TT(F^-_{i-1}), t(v_i), \gamma)$ and 
$\UnUpdate(\TT(F^+_{i+1}), t(v_{i+1}), \gamma)$.  
Since $\c{S}_1 = \{\TT(F^-_0), \TT(F^+_2), \TT(\bar{F}_1)\}$ is constructed 
in $O(n \log n)$ time, finding $\kappa$ by searching all $n$ tree sets 
takes $O(n \log^2 n)$ time.

\begin{theorem}
Given an undirected path $P = (V, A)$ and a $t \in \mathbb{R}^V$ together with a real $\gamma \geq 0$, the
$\gamma$-LUR of $t$ on $P$ can be found in $O(n\log^2 n)$ time.
\end{theorem}

\section{LIR on Rooted Trees} \label{sec:treelir} 

Let $T = (V, A)$ be a rooted tree with root $r$ and let for each vertex $v$, $T_v = (V_v, A_v)$ denote the subtree of $T$ rooted at $v$. Similar to the case of path LIR, for each vertex $v \in V$ we use the shorthands $E_v = E_{V_v}$ and $\tilde E_v = E_{T_v}^v$. Since the subtrees rooted at distinct children of a node $v$ are disjoint, 
one can  write an equation corresponding to (\ref{eq3}) in the case of paths, for any vertex $v$ of $T$:
\begin{equation}\label{eqte-tree}
\tilde E_{v}(x)  = (x - t(v))^2 +  \sum_{u \in \delta^-(v)}\min_{x - \gamma \leq y \leq x}\tilde E_{u}(y), 
\end{equation}
where $\delta^-(v) = \{u \in V \mid (u, v) \in A\}$.
An argument similar to that of Lemma \ref{lem:monotone} together with Lemma \ref{lem:convex-general} implies that 
for every $v \in V$, the function $\tilde E_v$ is convex and piecewise quadratic, and its derivative $F_v$ is continuous, monotonically increasing, and piecewise linear.
We can prove that $F_v$ satisfies the following recurrence where $\hF_u$ is defined analogously to $\hF_i$ in (\ref{eq:hatF}): 
\begin{equation}\label{eq9}
F_v(x) = 2(x - t(v)) + \sum_{u \in \delta^-(v)} \hF_{u}(x).
\end{equation}

Thus to solve the LIR problem on a tree, we post-order traverse the tree (from the leaves toward the root) and when processing a node $v$, we compute and sum up the linear functions $\hF_{u}$ for all children $u$ of $v$ and use the result in (\ref{eq9}) to compute the function $F_v$. We then solve $F_v(x) = 0$ to find $s^*_v$. As in the case of path LIR, $\hF_v$ can be represented by an ACT $\TT(\hF_v)$. For simplicity, we assume each non-leaf vertex $v$ has two children $h(v)$ and $l(v)$, where $|\hF_{h(v)}| \geq |\hF_{l(v)}|$.  We call $h(v)$ \emph{heavy} and $l(v)$ \emph{light}.  Given $\TT(\hF_{h(v)})$ and $\TT(\hF_{l(v)})$, we can compute $\TT(\hF_v)$ and $s_v^*$ with the operation $\Update(\TT(\hF_{h(v)} + \hF_{l(v)}), t(v), \gamma)$ in $O(|\hF_{l(v)}| (1 + \log |\hF_{h(v)}|/\log |\hF_{l(v)}|))$ time.  The merging of two functions dominates the time for the update.

\begin{theorem}\label{thm:LIRtree}
Given a rooted tree $T = (V, A)$ and a function $t \in \mathbb{R}^V$ together with a Lipschitz constant $\gamma$, one can find in $O(n \log n)$ time, the $\gamma$-Lipschitz isotonic regression of $t$ on $T$.
\end{theorem}
\begin{proof}
Let $T$ be rooted at $r$, it is sufficient to bound the cost of all calls to $\Update$ for each vertex: $U(T) = \sum_{v \in T} |\hF_{l(v)}| (1 + \log |\hF_{h(v)}| / \log |\hF_{l(v)}|)$.  

For each leaf vertex $u$ consider the path to the root $P_{u,r} = \langle u = v_1, v_2, \ldots, v_s = r\rangle$.  For each light vertex $v_i$ along the path, let its \emph{value} be $\nu(v_i) = 2+ \lceil \log(|\hF_{h(v_{i+1})}| / |\hF_{l(v_{i+1}) = v_i}|) \rceil$.  Also, let each heavy vertex $v_i$ along the path have \emph{value} $\nu(v_i) = 0$.  
For any root vertex $u$, the sum of values $\Upsilon_u = \sum_{v_i \in P_{u,r}} \nu(v_i)  \leq 3 \log n $, since $|\hF_{v_i}| 2^{\nu(v_i)-1} \leq |\hF_{v_{i+1}}|$ for each light vertex and $|\hF_r| = n$.  

Now we can argue that the contribution of each leaf vertex $u$ to $U(T)$ is at most $\Upsilon_u$.  Observe that for $v \in T$, then $\nu(l(v))$ is the same for any leaf vertex $u$ such that $l(v) \in P_{u,r}$.  Thus we charge $\nu(l(v))$ to $v$ for each $u$ in the subtree rooted at $l(v)$.  Since $\nu(l(v)) \geq 2 + \log(|\hF_{h(v)}|/|\hF_{l(v)}|) \geq 1 + \log |\hF_{h(v)}| / \log |\hF_{l(v)}|$ (for $\log |\hF_{h(v)}| \geq \log |\hF_{l(v)}|$), then the charges to each $v \in T$ is greater than its contribution to $U(T)$.  Then for each of $n$ leaf vertices $u$, we charge $\Upsilon_u \leq 3 \log n$.  Hence, $U(T) \leq 3 n \log n$.  
\end{proof}

\section{LUR on Trees}
\label{sec:treelur}

Let $T = (V,A)$ be an unrooted tree with $n$ vertices.  
Given any choice of $r \in V$ as a root, we say $T^{(r)}$ is the tree $T$ rooted at $r$.  
Given a function $t \in \b{R}^V$, a Lipschitz constraint $\gamma$, and a root $r \in V$, the $\gamma$-LIR regression of $t$ on $T^{(r)}$, denoted $s^{(r)}$, can be found in $O(n \log n)$ time using Theorem \ref{thm:LIRtree}.  We let $\xi^{(r)} = \min_{x \in \b{R}} \tilde{E}_r(x)$ as defined in (\ref{eqte-tree}) and let $r^* = \arg \min_{r \in V} \xi^{(r)}$.  The $\gamma$-Lipschitz unimodal regression ($\gamma$-LUR) of $t$ on $T$ is the function $s^{(r^*)} \in \b{R}^V$.  
Naively, we could compute the $\gamma$-LIR in $O(n^2 \log n)$ time by invoking Theorem \ref{thm:LIRtree} with each vertex as the root. In this section we show this can be improved to $O(n \log^3 n)$.

We first choose an arbitrary vertex $\hat{r} \in V$ (assume it has at most two edges in $A$) as an honorary root of $T$.  For simplicity only, we assume that every vertex of $T$ has degree three or less.  
Let the weight $\omega(v)$ of a subtree of $T^{(\hat{r})}$ at $v$ be the number of vertices below and including $v$ in $T^{(\hat{r})}$.
With respect to this honorary root, for each vertex $v \in V$, we define $p(v)$ to be the parent of $v$ and $h(v)$ (resp. $l(v)$) to be the heavy (resp. light) child of $v$ such that $\omega(h(v)) \geq \omega(l(v))$; $\hat{r}$ has no parent and leaf vertices (with respect to $\hat{r}$) have no children.  
If $v$ has only one child, label it $h(v)$. 
If $v = l(p(v))$ we say $v$ is a \emph{light} vertex, otherwise it is a \emph{heavy} vertex.
The following two properties are consequences of $\omega(h(v)) > \omega(l(v))$ for all $v \in T^{(\hat{r})}$.
\begin{itemize}
\vspace{-.1in}
\item[(P1)] $\sum_{v \in V} \omega(l(v)) = O(n \log n)$.  
\vspace{-.1in}
\item[(P2)] For $u \in V$, the path $\langle \hat{r}, \ldots, u \rangle$ in $T^{(\hat{r})}$ contains at most $\lceil\log_2 n\rceil$ light vertices.
\end{itemize}

As a preprocessing step, we construct the $\gamma$-LIR of $t$ for $T^{(\hat{r})}$ as described in Theorem \ref{thm:LIRtree}, and, for each light vertex $v$, we store a copy of $\TT(\hF_v)$ before merging with $\TT(\hF_{h(p(v))})$.  The total size of all stored ACTs is $O(n \log n)$ by property (P1).  
We now perform an inorder traversal of $T^{(\hat{r})}$, letting each $v \in V$ in turn be the structural root of $T$ so we can determine $\xi^{(v)}$.  
As we traverse, when $v$ is structural root, we maintain a tree set $\c{S}(\hF_{p(v)})$ at $p(v)$ with at most $k = \log_2 n$ trees, and ACTs $\c{T}(\hF_{h(v)})$ and $\c{T}(\hF_{l(v)})$ at $h(v)$ and $l(v)$, respectively.  
All functions $\hF_u$ for $u = \{p(v), l(v), h(v)\}$ are defined the same as with rooted trees, when rooted at the structural root $v$, for the subtree rooted at $u$.   
Given these three data structures we can compute $\xi_v$ in $O(\log^3 n)$ time by first temporarily inserting the ACTs $\TT(\hF_{h(v)})$, $\TT(\hF_{l(v)})$, and $\TT(\bar{F}_v)$ into $\c{S}(\hF_{p(v)})$, thus describing the function $F_v$ for a rooted tree with $v$ as the root, and then performing an $s^{(v)} = \evali(0)$ and $\integrate(s^{(v)})$ on $\c{S}(F_v)$.  Recall $\bar{F}_v(x) = 2(x-t(v))$.  
Thus to complete this algorithm, we need to show how to maintain these three data structures as we traverse $T^{(\hat{r})}$.  

During the inorder traversal of $T^{(\hat{r})}$ we need to consider 3 cases: letting $h(v)$, $l(v)$, or $p(v)$ be the next structural root, as shown in Figure \ref{fig:traverse}.  
\begin{itemize}
\item[(a)] $h(v)$ is the next structural root: 
To create $\c{S}(\hF_{p(h(v))}) = \c{S}(\hF_v)$ we add $\TT(\hF_{l(v)}))$ to the tree set $\c{S}(\hF_{p(v)})$ through a $\merge(\c{S}(\hF_{p(v)}), \TT(\hF_{l(v)}))$ operation and perform $\Update(\c{S}(\hF_{p(v)}), t(v), \gamma)$.  
The ACT $\TT(\hF_{l(h(v))})$ is stored at $l(h(v))$.  
To create the ACT $\TT(\hF_{h(h(v))})$ we perform $\UnUpdate(\TT(\hF_{h(v)}),\, t(h(v)), \gamma)$, and then set $\TT(\hF_{h(h(v))}) \leftarrow \TT(\hF_{h(h(v))} - \hF_{l(h(v))})$, using $\TT(\hF_{l(h(v))})$.  

\item[(b)] $l(v)$ is the next structural root: 
To create $\c{S}(\hF_{p(l(v))})$ we perform $\include(\c{S}(\hF_{p(v)}), \c{T}(\hF_{h(v)}))$ and then $\Update(\c{S}(\hF_{p(v)}), t(v), \gamma)$.  
The ACTs for $h(l(v))$ and $l(l(v))$ are produced the same as above.  

\item[(c)] $p(v)$ is the next structural root: 
To create $\TT(\hF_v)$ we set $\TT(\hF_v) \leftarrow \TT(\hF_{h(v)} + \hF_{l(v)})$ and then perform $\Update(\TT(\hF_v), t(v), \gamma)$.
To create $\c{S}(\hF_{p(p(v))})$ we first set $\c{S}(\hF_{p(p(v))}) \leftarrow \UnUpdate(\c{S}(\hF_{p(v)}), t(p(v)), \gamma)$.  
Then if $v = h(p(v))$, we perform $\unmerge(\c{S}(\hF_{p(p(v))}), \c{T}(\hF_{l(p(v))}))$, and $\TT(\hF_{l(p(v))})$ is stored at $l(p(v))$.
Otherwise $v = l(p(v))$, and we perform $\uninclude(\c{S}(\hF_{p(p(v))}), \c{T}(\hF_{h(p(v))}))$ to get $\c{S}(\hF_{p(p(v))})$, and $\c{T}(\hF_{h(p(v))})$ is the byproduct.  
\end{itemize}

\begin{figure}
\begin{center}
\includegraphics[scale=1]{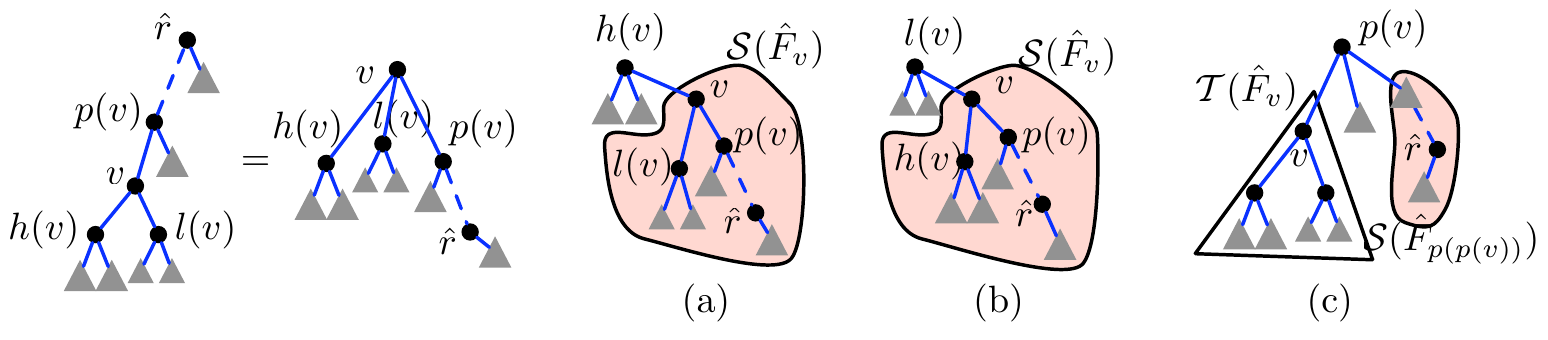}
\end{center}
\myCaption{\label{fig:traverse} Illustration of $T^{(\hat{r})}$ ordered with respect to honorary root $\hat{r}$.  This is redrawn, ordered with respect to  structural root as $v$.  The tree is redrawn with structural root as $h(v)$ after case (a), $l(v)$ after case (b), and $p(v)$ after case (c).}
\end{figure}

We observe that on a traversal through any vertex $v \in V$, the number of ACTs in the tree set $\c{S}(\hF_{p(v)})$ only increases when $l(v) = u$ is the new structural root.  Furthermore, when the traversal reverses this path so that $p(u)$ is the new root where $u = l(p(u))$, the same ACT is removed from the tree set.  Thus, when any vertex $v \in V$ is the structural root, we can bound the number of ACTs in the tree set $\c{S}(\hF_{p(v)})$ by the number of light vertices on the path from $\hat{r}$ to $v$.  Thus property (P2) implies that $\c{S}(\hF_{p(v)})$ has at most $\log_2 n$ ACTs.  

We can now bound the runtime of the full traversal algorithm described above.  Each vertex $v \in V$ is visited as the structural root $O(1)$ time and each visit results in $O(1)$ $\Update$, $\UnUpdate$, $\include$, $\uninclude$, $\merge$, and $\unmerge$ operations.  
Performing a single $\Update$ or $\UnUpdate$ on a tree set takes at most $O(k \log^2 n) = O(\log^3 n)$ time, and each $\include$ and $\uninclude$ operation takes $O(1)$ time, so the total time is $O(n \log^3 n)$.  
We now need to bound the global costs of all $\merge$ and $\unmerge$ operations.  In proving Theorem \ref{thm:LIRtree} we showed that the cost of performing a $\merge$ of all light vertices into heavy vertices as we build $\c{T}(\hF_{\hat{r}})$ takes $O(n \log^2 n)$.  Since at each vertex $v \in V$ we only $\merge$ and $\unmerge$ light vertices, the total time of all $\merge$ and $\unmerge$ operations is again $O(n \log^2 n)$.  

Thus we can construct $\xi^{(v)}$ for each possible structural root, and using $r^* = \arg \min_{v \in V} \xi^{(v)}$ as the optimal root we can invoke Theorem \ref{thm:LIRtree} to construct the $\gamma$-LIR of $t$ on $T^{(r^*)}$.  

\begin{theorem}
Given an unrooted tree $T = (V,A)$ and a function $t \in \b{R}^V$ together with a Lipschitz constraint $\gamma$, we can find in $O(n \log^3 n)$ time 
the $\gamma$-Lipschitz unimodal regression of $t$ on $T$.
\label{thm:LURtree}
\end{theorem}

\section{Conclusion}\label{sec:conclusion}
In this paper we defined the Lipschitz isotonic/unimodal regression 
problems on DAGs and provided efficient algorithms for solving it on the special cases where the graph is a path or a tree. 
Our ACT data structure has proven to be a powerful tool in our approach.


Our algorithms can be generalized in a number of ways, as
listed below, without affecting their asymptotic running times. 
\begin{itemize} 
\item One can specify different Lipschitz value $\gamma$ for each Lipschitz constraint, thus writing the Lipschitz constraint involving $s_i$ and $s_{i+1}$ as $s_{i+1} \leq s_i + \gamma_i$. 
\item Isotonicity constraints can be regarded as (backward) Lipschitz constraints for which the Lipschitz constant is zero. One can allow this constant to be chosen arbitrarily and indeed permit different constants for different isotonicity constraints. In other words, the constraint $s_i \leq s_{i+1}$ can be replaced with $s_{i+1} \geq s_i + \delta_i$.

\item The $\ell_2$ norm $\|\cdot\|$ can be replaced with a ``weighted'' version $\|\cdot\|_{\llambda}$ for any $\lambda: V \rightarrow \mathbb{R}$ by defining for a function $s: V \rightarrow \mathbb{R}$:  
\[
\|s\|^2_{\lambda} = \sum_{v \in V} \lambda(v)s(v)^2.
\]
\end{itemize}

The most important open problem is solving LIR problem on general DAGs. The known algorithm for solving IR on DAGs runs in $O(n^4)$ and does not lend itself to the Lipschitz generalization. The difficulty here is to compute the function $\tilde E_i$ (or in fact $F_i$) when the $i$th vertex is processed and comes from the fact that the $F_{i_j}$ function, $j = 1, \dots, k$, where $i_1, \dots, i_k$ are vertices with incoming edges to $i$, are not independent. Independently, it is an important question whether IR can be solved more efficiently. In particular, the special case of the problem where the given DAG is planar has important applications in terrain simplification. 

\subsection*{Acknowledgements}
We anonymous reviewers for many helpful remarks that helped improve the presentation of this paper and we thank G\"unter Rote for suggesting the improvement in the analysis of $\TT(F+G)$, similar to his paper~\cite{Rote}.  

\bibliographystyle{abbrv}
\bibliography{lir}

\begin{thebibliography}{10}

\bibitem{Agarwal:06}
P.~K. Agarwal, L.~Arge, and K.~Yi.
\newblock {I}/{O}-efficient batched union-find and its applications to terrain
  analysis.
\newblock In {\em Proceedings 22nd ACM Symposium on Computational Geometry},
  pages 167--176, 2006.

\bibitem{AHKW06}
S.~Angelov, B.~Harb, S.~Kannan, and L.-S. Wang.
\newblock Weighted isotonic regression under the $l_1$ norm.
\newblock In {\em Proc. 17th Annual ACM-SIAM Symposium on Discrete Algorithms},
  pages 783--791, 2006.

\bibitem{Attali:08}
D.~Attali, M.~Glisse, S.~Hornus, F.~Lazarus, and D.~Morozov.
\newblock Persistence-sensitive simplification of functions on surfaces in
  linear time.
\newblock Manuscript, INRIA, 2008.

\bibitem{ABERS55}
M.~Ayer, H.~D. Brunk, G.~M. Ewing, W.~T. Reid, and E.~Silverman.
\newblock An empirical distribution function for sampling with incomplete
  information.
\newblock {\em Annals of Mathematical Statistics}, 26:641--647, 1955.

\bibitem{Bajaj:98}
C.~Bajaj, V.~Pascucci, and D.~Schikore.
\newblock Visualization of scalar topology for structural enhancement.
\newblock In {\em IEEE Visualizationw}, pages 51--58, 1998.

\bibitem{BBBB72}
R.~E. Barlow, D.~J. Bartholomew, J.~M. Bremmer, and H.~D. Brunk.
\newblock {\em Statistical Inference Under Order Restrictions: The Theory and
  Application of Isotonic Regression}.
\newblock Wiley Series in Probability and Mathematical Statistics. John Wiley
  and Sons, 1972.

\bibitem{Bremer:03}
P.~Bremer, H.~Edelsbrunner, B.~Hamann, and V.~Pascucci.
\newblock A multi-resolution data structure for two-dimensional morse
  functions.
\newblock In {\em Proceedings 14th IEEE Visualization Conference}, pages
  139--146, 2003.

\bibitem{BT79}
M.~R. Brown and R.~E. Tarjan.
\newblock A fast merging algorithm.
\newblock {\em Journal of Algorithms}, 15:416--446, 1979.

\bibitem{Bru55}
H.~D. Brunk.
\newblock Maximum likelihood estimates of monotone parameters.
\newblock {\em Annals of Mathematical Statistics}, 26:607--616, 1955.

\bibitem{Rote}
K.~Buchin, S.~Cabello, J.~Gudmendsson, M.~L\"offler, J.~Luo, G.~Rote, R.~I.
  Silveira, B.~Speckmann, and T.~Wolle.
\newblock Finding the most relevant fragments in networks.
\newblock Manuscript, October 2009.

\bibitem{Danner:07}
A.~Danner, T.~M\o{}lhave, K.~Yi, P.~K. Agarwal, L.~Arge, and H.~Mitasova.
\newblock Terrastream: from elevation data to watershed hierarchies.
\newblock In {\em 15th ACM International Symposium on Advances in Geographic
  Information Systems}, 2007.

\bibitem{Edelsbrunner:07}
H.~Edelsbrunner and J.~Harer.
\newblock {\em Persistent Homology: A Survey}.
\newblock Number 453 in Contemporary Mathematics. American Mathematical
  Society, 2008.

\bibitem{Edelsbrunner:03}
H.~Edelsbrunner, J.~Harer, and A.~Zomorodia.
\newblock Hierarchical morse - smale complexes for piecewise linear
  2-manifolds.
\newblock {\em Discrete {\&} Computational Geometry}, 30(1):87--107, 2003.

\bibitem{Edelsbrunner:00}
H.~Edelsbrunner, D.~Letscher, and A.~Zomorodian.
\newblock Topological persistence and simplification.
\newblock In {\em Proceedings 41st Annual Symposium on Foundations on Computer
  Science}, pages 454--463, 2000.

\bibitem{Edelsbrunner:06}
H.~Edelsbrunner, D.~Morozov, and V.~Pascucci.
\newblock Persistence-sensitive simplification functions on 2-manifolds.
\newblock In {\em Proceedings 22nd ACM Symposium on Computational Geometry},
  pages 127--134, 2006.

\bibitem{Frederickson:82}
G.~N. Frederickson and D.~B. Johnson.
\newblock The complexity of selection and ranking in $x+y$ and matrices with
  sorted columns.
\newblock {\em J. Comput. Syst. Sci.}, 24(2):192--208, 1982.

\bibitem{GW84}
S.~J. Grotzinger and C.~Witzgall.
\newblock Projection onto order simplexes.
\newblock {\em Applications of Mathematics and Optimization}, 12:247--270,
  1984.

\bibitem{Guskov:01}
I.~Guskov and Z.~J. Wood.
\newblock Topological noise removal.
\newblock In {\em Graphics Interface}, pages 19--26, 2001.

\bibitem{Jewel:75}
W.~S. Jewel.
\newblock Isotonic optimization in tariff construction.
\newblock {\em {ASTIN} Bulletin}, 8(2):175--203, 1975.

\bibitem{MM85}
W.~L. Maxwell and J.~A. Muckstadt.
\newblock Establishing consistent and realistic reorder intervals in
  production-distribution systems.
\newblock {\em Operations Research}, 33:1316--1341, 1985.

\bibitem{Ni:04}
X.~Ni, M.~Garland, and J.~C. Hart.
\newblock Fair {M}orse functions for extracting the topological structure of a
  surface mesh.
\newblock {\em ACM Transactions Graphaphics}, 23:613--622, 2004.

\bibitem{PX99}
P.~M. Pardalos and G.~Xue.
\newblock Algorithms for a class of isotonic regression problems.
\newblock {\em Algorithmica}, 23:211--222, 1999.

\bibitem{Preparata:85}
F.~P. Preparata and M.~I. Shamos.
\newblock {\em Computational geometry: an introduction}.
\newblock Springer-Verlag, New York, NY, USA, 1985.

\bibitem{RWD88}
T.~Robertson, F.~T. Wright, and R.~L. Dykstra.
\newblock {\em Order Restricted Statistical Inference}.
\newblock Wiley Series in Probability and Mathematical Statistics. John Wiley
  and Sons, 1988.

\bibitem{Soille:04a}
P.~Soille.
\newblock Morphological carving.
\newblock {\em Pattern Recognition Letters}, 25:543--550, 2004.

\bibitem{Soille:04}
P.~Soille.
\newblock Optimal removal of spurious pits in grid digital elevation models.
\newblock {\em Water Resources Research}, 40(12), 2004.

\bibitem{Soille:03a}
P.~Soille, J.~Vogt, and R.~Cololmbo.
\newblock Carbing and adaptive drainage enforcement of grid digital elevation
  models.
\newblock {\em Water Resources Research}, 39(12):1366--1375, 2003.

\bibitem{SWW03}
J.~Spouge, H.~Wan, and W.~J. Wilbur.
\newblock Least squares isotonic regression in two dimensions.
\newblock {\em Journal of Optimization Theory and Applications}, 117:585--605,
  2003.

\bibitem{Sto00}
Q.~F. Stout.
\newblock Optimal algorithms for unimodal regression.
\newblock {\em Computing Science and Statistics}, 32:348--355, 2000.

\bibitem{Stout:08}
Q.~F. Stout.
\newblock $l_{\infty}$ isotonic regression.
\newblock Unpublished Manuscript, 2008.

\bibitem{Sto08}
Q.~F. Stout.
\newblock Unimodal regression via prefix isotonic regression.
\newblock {\em Computational Statistics and Data Analysis}, 53:289--297, 2008.

\bibitem{Tho62}
W.~A. {Thompson, Jr.}
\newblock The problem of negative estimates of variance components.
\newblock {\em Annals of Mathematical Statistics}, 33:273--289, 1962.

\bibitem{vKS09}
M.~J. van Kreveld and R.~I. Silveira.
\newblock Embedding rivers in polyhedral terrains.
\newblock In {\em Proceedings 25th Symposium on Computational Geometry}, 2009.

\bibitem{Zomorodian:09}
A.~Zomorodian.
\newblock Computational topology.
\newblock In M.~Atallah and M.~Blanton, editors, {\em Algorithms and Theory of
  Computation Handbook}, page (In Press). Chapman \& Hall/CRC Press, Boca
  Raton, FL, second edition, 2009.

\bibitem{Zomorodian:05}
A.~Zomorodian and G.~Carlsson.
\newblock Computing persistent homology.
\newblock {\em Discrete \& Computational Geometry}, 33(2):249--274, 2005.

\end{thebibliography}

\newpage

\appendix

\section{$\integrate$ Operation}

\begin{lemma} 
We can modify an ACT datastructure for  $\TT(F)$ with $|F| =n$ and $F(x) = d E(x) / dx$, so for any $x \in R$ we can calculate $E(x) = \integrate(x)$ in $O(\log n)$ time, 
and so the cost of maintaining $\TT(F)$ is asymptotically equivalent to without this feature.
\label{lem:integrate}
\end{lemma}
\begin{proof}
For any $a \in \b{R}$, we calculate $\integrate(a) = E(a)$ by adding $E(b_1)$ to $\int_{b_1}^a F(x) \, dx$, where $b_1$ is the leftmost breakpoint in $F$.  Maintaining $E(b_1)$ is easy and is explained first.  Calculating $\int_{b_1}^a F(x) \, dx$ requires storing partial integrals and other information at each node of $\TT(F)$ and is more involved.  

We explain how to maintain the value $E(b_1)$ in the context of our $\gamma$-LIR problems in order to provide more intuition, but a similar technique can be used for more general ACTs.  
We maintain values $c_1, c_2, c_3$ such that $c_1 x^2 + c_2 x + c_3 = \sum_{i=1}^k (x-t_i)^2$ for the $k=O(n)$ values of $t_i$.  This is done by adding $1$, $-2t_i$, and $t_i^2$ to $c_1$, $c_2$, and $c_3$ every time a new $t_i$ is processed.  Or if two trees are merged, this can be updated in $O(1)$ time.  The significance of this equation is that $E(b_1) = c_1 b_1^2 + c_2 b_1 + c_3$ because the isotonic condition is enforced for all edges in $A$, forcing all values $s_i$ in the $\gamma$-LIR to be equal.

When we construct $\TT(F)$ we keep the following extra information at each node: 
 each node $z$ whose subtree spans breakpoints $b_i$ to $b_j$ maintains the integral $I(z) = \int_{b_i}^{b_j} \Phi_z^{-1}(F(x)-F(b_i))\, dx$ (this is done for technical reasons; subtracting $F(b_i)$ ensures $F(x)-F(b_i) \geq 0$ for $x \in [b_i, b_j]$ since $F$ is isotonic), the width $W(z) = \Phi_z^{-1}(b_j) - \Phi_z^{-1}(b_i)$, and the height $H(z) = \Phi_z^{-1}(F(b_j)) - \Phi_z^{-1}(F(b_i))$.  
Recall that $\Phi_z$ is the composition of affine transformations stored at the nodes along the path from the root to $z$.  
We build these values ($I(z)$, $W(z)$, and $H(z)$ for each $z \in \TT(F)$) from the bottom of $\TT(F)$ up, so we apply all appropriate affine transformations $\phi_{z_j}$ for decedents $z_j$ of $z$, but not yet those of ancestors of $z$.  In particular, if node $z$ has left child $z_l$ and right child $z_k$, then we can calculate 
\begin{eqnarray*}
I(z) &=& \phi_{z_l}(I(z_l)) + \phi_{z_k}(I(z_k)) + \phi_{z_l}(H(z_l)) \cdot \phi_{z_k}(W(z_k)) \textrm{ (see Figure \ref{fig:rotate}), } 
\\
H(z) &=& \phi_{z_l}(H(z_l)) + \phi_{z_k}(H(z_k)), \textrm{ and } 
\\
W(z) &=& \phi_{z_l}(W(z_l)) + \phi_{z_k}(W(z_k)).
\end{eqnarray*}  
When a new breakpoint $b$ is inserted or removed there are $O(\log n)$ nodes that are either rotated or are on the path from the node storing $b$ to the root.  For each affected node $z$ we recalculate $I(z)$, $W(z)$ and $H(z)$ in a bottom up manner using $z$'s children once they have been updated themselves.  Likewise, tree rotations require that $O(1)$ nodes are updated.  So to demonstrate that the maintenance cost does not increase, we just need to show how to update $I(z)$, $H(z)$ and $W(z)$ under an affine transformation $\phi_z(z)$ in $O(1)$ time.  

For affine transformation $\phi_z$ applied at node $z$ (spanning breakpoints $b_i$ through $b_j$), we can decompose it into three components: translation, scaling, and rotation.  
The translation does not affect $I(z)$, $W(z)$, or $H(z)$.  
The scaling can be decomposed into scaling of width $w$ and of height $h$.  We can then update $W(z) \leftarrow w \cdot W(z)$, $H(z) = h \cdot H(z)$, and $I(z) = w\cdot h\cdot I(z)$.  
To handle rotation, since translation does not affect $I(z)$, $W(z)$, or $H(z)$, we consider the data translated by $\Phi_z^{-1}((-b_i, 0))$ so that $\Phi_z^{-1}((b_i, F(b_i)-F(b_i))) = (0,0)$ is the origin and $\Phi_z^{-1}((b_j, F(b_j)-F(b_i))) = (x,y) = (W(z), H(z))$.  Then we consider rotations about the origin, as illustrated in Figure \ref{fig:rotate}.  
After a rotation by a positive angle $\theta$, the point $(x,y)$ becomes $(x', y') = (x \cos \theta - y \sin \theta, x \sin \theta + y \cos \theta)$.  This updates $W(z) \leftarrow x'$ and $H(z) \leftarrow y'$.  We can update $I(z)$ by subtracting the area $A_-$ that is no longer under the integral $y \cdot (y \tan \theta)/2$ and adding the area $A_+$ newly under the integral $x' \cdot (x' \tan \theta)/2$ as illustrated in Figure \ref{fig:rotate}.  Thus $I(z) \leftarrow I(v) + x' \cdot (x' \tan \theta)/2 - y \cdot (y \tan \theta)/2$.  A similar operation exists for a negative angle $\theta$, using the fact that $F$ must remain monotone.
As desired all update steps take $O(1)$ time.

\begin{figure}[h]
\begin{center}
\includegraphics{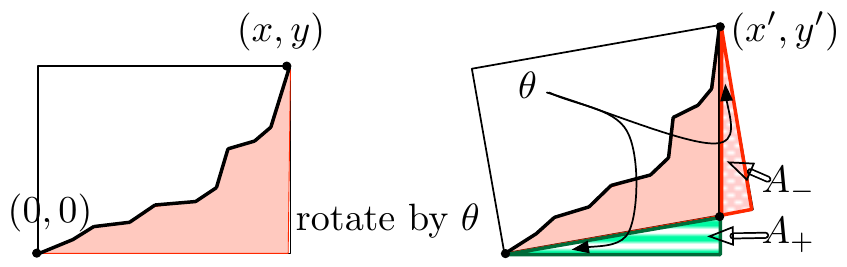} 
\hspace{.8in}
\includegraphics{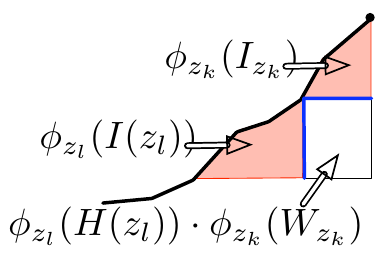}
\end{center}
\myCaption{\label{fig:rotate} Left: Illustration of the change in integral caused by a rotate of $\theta$ degrees counterclockwise. 
Right: Calculation of $I_z$ from children $z_k$ and $z_l$ of $z$.}
\end{figure}

Now for a value $a \in \b{R}$ we can calculate $\integrate(a) = E(a)$ as follows.  
Let $z_a \in \TT(F)$ be the node and $b$ be the breakpoint associated with $\pred_x(a)$.  We consider the path $Z = \langle z_a, \ldots, z_1, z_0\rangle$ from $z_a$ to the root $z_0$ of $\TT(F)$, and will inductively build an integral along $Z$.  
We calculate $I_{z_0} = \int_{b_1}^a \phi_{z_0}(F(x) - F(b_1)) \, dx$ similar to the way we calculated $I(z_0)$, but just over the integral $[b_1, a]$.  
As a base case, let $I_{z_a} = \int_{b}^{a} \Phi_{z_a}^{-1}(F(x)-F(b)) \, dx$ (which is easy to calculate because $F$ is linear in this range) and let $W_{z_a} = \Phi_{z_a}^{-1}(a)-\Phi_{z_a}^{-1}(b)$ describe the integral and width associated with the range $[b,a]$ and $z_a$.  
For a node $z \in Z$ that spans breakpoints $b_i$ through $b_j$, we calculate an integral $I_z = \int_{b_i}^a \Phi_z^{-1}(F(x)-F(b_i)) \, dx$ and width $W_z = \Phi_z^{-1}(a) - \Phi_z^{-1}(b_i)$ using its two children: $z_l$ and $z_k$.  Let $z_k$ be the right child which lies in $Z$ and for which we had inductively calculated $I_{z_k}$ and $W_{z_k}$.  
Then let $I_z = \phi_{z_l}(I(z_l)) + \phi_{z_k}(I_{z_k}) + \phi_{z_l}(H(z_l)) \cdot \phi_{z_k}(W_{z_k})$ (as illustrated in Figure \ref{fig:rotate}) and let $W_z = \phi_{z_l}(W(z_l)) + \phi_{z_k}(W_{z_k})$.  
Once $\phi_{z_0}(I_{z_0})$ and $\phi_{z_0}(W_{z_0})$ are calculated at the root, we need to adjust for the $F(b_1)$ term in the integral.  We add $\phi_{z_0}(W_{z_0})\cdot F(b_1)$ to $\phi_{z_0}(I(z_0)) = \int_{b_1}^a (F(x) - F(b_1)) \, dx$ to get $\int_{b_1}^a F(x) \, dx$.
So finally we set $\integrate(a) =  E(b_1) + \phi_{z_0}(I_{z_0}) + \phi_{z_0}(W_{z_0}) \cdot F(b_1)$.  
Since the path $Z$ is of length at most $O(\log n)$, $E(a)$ can be calculated in $O(\log n)$ time. 
\end{proof}

\end{document}